\documentclass[11pt,a4paper]{amsart}
\usepackage[foot]{amsaddr}
\usepackage{ifxetex}
\ifxetex
  \usepackage[no-math]{fontspec}
\else
\fi
\usepackage{amsmath}
\usepackage{amsfonts}
\usepackage{amssymb}
\usepackage{amsthm}
\usepackage{fullpage}
\usepackage{microtype}
\ifxetex 
  \usepackage[libertine]{newtxmath}
\else
  \usepackage{newtxmath}
\fi
\usepackage[tt=false]{libertine} 
\usepackage{bm}
\usepackage{bbm}
\usepackage{mathtools} 
\usepackage{algorithm}
\usepackage{algpseudocode}
\usepackage{enumitem}

\usepackage[margin=1cm]{caption} 
\usepackage{subfig}

\usepackage[thinlines]{easytable}

\usepackage[bookmarks=true,hypertexnames=false,pagebackref]{hyperref}
\hypersetup{colorlinks=true, citecolor=blue, linkcolor=red, urlcolor=blue}

\usepackage{tikz}
\usetikzlibrary{arrows,arrows.meta,backgrounds,calc,fit,decorations.pathreplacing,decorations.markings,shapes.geometric}

\tikzstyle{internal} = [draw, fill, shape=circle]
\tikzstyle{external} = [shape=circle]
\tikzstyle{square}   = [draw, fill, rectangle]
\tikzstyle{triangle} = [draw, fill, regular polygon, regular polygon sides=3, inner sep=3pt]
\tikzstyle{pentagon} = [draw, fill, regular polygon, regular polygon sides=5, inner sep=2pt, minimum size=14pt]
\tikzset{every fit/.append style=text badly centered}

\usetikzlibrary{positioning,chains,fit,shapes,calc}
\usetikzlibrary{trees}
\usetikzlibrary{decorations.pathreplacing}
\usetikzlibrary{decorations.pathmorphing}
\usetikzlibrary{decorations.markings}
\tikzset{>=latex} 

\usepackage{ifthen}

\usepackage{cleveref}

\usepackage[textsize=tiny]{todonotes}

\usepackage[normalem]{ulem}

\usepackage{mleftright}

\def\*#1{\mathbf{#1}}
\def\+#1{\mathcal{#1}}
\def\-#1{\mathrm{#1}}
\def\=#1{\mathbb{#1}}

\newcommand{\NP}{{\mathbf{NP}}}

\newcommand{\wt}{\mathrm{wt}}

\newtheorem{theorem}{Theorem}

\newtheorem{lemma}[theorem]{Lemma}

\crefname{theorem}{Theorem}{Theorems}
\crefname{observation}{Observation}{Observations}
\crefname{claim}{Claim}{Claims}
\crefname{condition}{Condition}{Conditions}
\crefname{algorithm}{Algorithm}{Algorithms}
\crefname{property}{Property}{Properties}
\crefname{example}{Example}{Examples}
\crefname{fact}{Fact}{Facts}
\crefname{lemma}{Lemma}{Lemmas}
\crefname{corollary}{Corollary}{Corollaries}
\crefname{definition}{Definition}{Definitions}
\crefname{remark}{Remark}{Remarks}
\crefname{proposition}{Proposition}{Propositions}
\crefname{equation}{equation}{equations}
\crefname{enumi}{Case}{Case}
\creflabelformat{enumi}{(#2#1#3)}

\makeatletter
\def\prob#1#2#3{\goodbreak\begin{list}{}{\labelwidth\z@ \itemindent-\leftmargin
      \itemsep\z@  \topsep6\p@\@plus6\p@
      \let\makelabel\descriptionlabel}
  \item[\textbf{Name}]#1
  \item[\textbf{Instance}]#2
  \item[\textbf{Output}]#3
  \end{list}}
\makeatother

\makeatletter
\providecommand\@dotsep{5}
\def\listtodoname{Todo list}
\def\listoftodos{\@starttoc{tdo}\listtodoname}
\makeatother

\title{Inapproximability of Counting Independent Sets in Linear Hypergraphs}

\author{Guoliang Qiu, Jiaheng Wang}
\address[Guoliang Qiu]{John Hopcroft Center for Computer Science, Shanghai Jiao Tong University, 800 Dongchuan Road, Minhang District, Shanghai, China. \textnormal{E-mail: \url{guoliang.qiu@sjtu.edu.cn}}}
\address[Jiaheng Wang]{School of Informatics, University of Edinburgh, Informatics Forum, Edinburgh, EH8 9AB, United Kingdom. \textnormal{E-mail: \url{jiaheng.wang@ed.ac.uk}}}

\begin{document}

\begin{abstract}
It is shown in this note that approximating the number of independent sets in a $k$-uniform linear hypergraph with maximum degree at most $\Delta$ is $\NP$-hard if $\Delta\geq 5\cdot 2^{k-1}+1$. 
This confirms that for the relevant sampling and approximate counting problems, the regimes on the maximum degree where the state-of-the-art algorithms work are tight, up to some small factors. 
These algorithms include: the approximate sampler and randomised approximation scheme by Hermon, Sly and Zhang (RSA, 2019), the perfect sampler by Qiu, Wang and Zhang (ICALP, 2022), and the deterministic approximation scheme by Feng, Guo, Wang, Wang and Yin (FOCS, 2023). 
\end{abstract}

\maketitle

\section{Introduction}


This note is concerned with the problem of counting independent sets in hypergraphs. 
We start with the basic definitions. 
A \emph{hypergraph} $H=(W,\mathcal{E})$ is specified by a set of vertices $W$ and a set of hyperedges $\+E$, where each hyperedge $e\in\mathcal{E}$ is a subset of $W$. 
It is said to be $k$-\emph{uniform}, if each hyperedge contains exactly $k$ vertices. 
The degree of a vertex is the number of hyperedges in which it appears, 
and the \emph{degree} $\Delta$ of the hypergraph is the maximum degree of its vertices. 
A set $I\subseteq W$ is a (weak) \emph{independent set} if $I\cap e\neq e$ holds for all $e\in\mathcal{E}$. 
We remark that this problem is parameterised by $k$ and $\Delta$. 

This problem was first studied by Bordewich, Dyer and Karpinski \cite{BDK06,BDK08}, where they showed that the straight-forward Markov chain over all independent sets is rapid mixing when $\Delta\leq k-2$. 
This also yields a fully-polynomial randomised approximation scheme (FPRAS) for counting the number of independent sets by the standard counting-to-sampling reductions. For more information on such reductions, see for example \cite{JVV86} or \cite{SVV09}.
The regime on the maximum degree was further improved by Hermon, Sly and Zhang \cite{HSZ19} to $\Delta\leq c\cdot 2^{k/2}$ for some absolute constant $c>0$. 
A later work by Qiu, Wang and Zhang \cite{qiu2022perfect} provided a perfect sampler (i.e., the output distribution is unbiased) which runs in expected polynomial time when $\Delta\leq c\cdot 2^{k/2}/k$ for some absolute constant $c>0$. 
Very recently, Feng, Guo, Wang, Wang and Yin \cite{FGWWY22} further derandomised the Markov chain Monte Carlo approach and provided a fully-polynomial deterministic approximation scheme (FPTAS) when $\Delta\leq c\cdot 2^{k/2}/k^2$ for some absolute constant $c>0$. 
On the other hand, Bez\'{a}kov\'{a}, Galanis, Goldberg, Guo and \v{S}tefankovi\v{c} \cite{BGGGS19} proved that approximating the number of independent sets is intractable when $\Delta\geq 5\cdot 2^{k/2}$ unless $\mathbf{NP}=\mathbf{RP}$. 
These results established a sharp \emph{computational phase transition} for this problem.

The notion of \emph{linear} hypergraphs (aka. \emph{simple} hypergraphs) has also attracted some attention. 
A hypergraph is considered linear if the intersection of any two hyperedges contains at most one vertex. 
When the input hypergraph is restricted to be linear, the tractability regime of the aforementioned algorithms can be extended. 
Specifically, the same work mentioned above also established the following regimes: $\Delta\leq c\cdot 2^{k}/k^2$ for both the FPRAS and the perfect sampler, and $\Delta\leq 2^{(1-o(1))k}$ for the FPTAS.
Under this setting, can we also establish a computational phase transition up to some low-order factors? 
As the main claim of this note, we answer the question affirmatively.
Our result holds for a larger class of hypergraphs with overlap $b$, where the intersection of any two hyperedges contains at most $b$ vertices. 
\begin{theorem} \label{thm:main}
  For any $k\geq 2$, $1\leq b\leq k/2$ and $\Delta\geq 5\cdot 2^{k-b}+1$, it is $\NP$-hard to approximate the number of independent sets in $k$-uniform hypergraphs with maximum degree at most $\Delta$ and overlap at most $b$. 
\end{theorem}

The hardness for the linear hypergraph can be obtained by setting $b=1$ in the theorem. Additionally, it subsumes the general case presented in \cite{BGGGS19} by setting $b:=\lfloor k/2\rfloor$. 


In the general case \cite{BGGGS19}, the hardness can be shown by reducing from the hard-core model (counting weighted independent sets) on \emph{graphs}, where each vertex in the graph is replaced by $k/2$ copies in the resulting hypergraph, and each edge in the graph is associated with the hyperedge containing these copies. However, this approach naturally requires large overlaps in the resulting hypergraph.
In our case, we ensure linearity (or small overlaps) by creating $b$ copies for each vertex in the graph, and then filling up each hyperedge to $k$ vertices. 
This boils down to a general anti-ferromagnetic $2$-spin system, instead of the hard-core model.
The main complication is to establish the so-called ``non-uniqueness'' property, which allows us to utilize a classical result by Sly and Sun~\cite{SS14} to derive inapproximability. 

An open problem is to determine the exact computational phase transition for the linear case, as there remains a $\Theta(k^2)$ gap in-between.
Notably, for independent sets on the infinite linear hypertree, the uniqueness of the Gibbs distribution holds when $\Delta\leq \left(1+o(1)\right)2^{k}/k$~\cite{BGGGS19}.
It is still unclear which of these three thresholds, if any, would be the ground truth for the computational phase transition point. 

The phenomenon of a more relaxed algorithmic regime (and thus a more restricted hardness regime) existing for linear hypergraphs is also observed in the context of hypergraph $q$-colouring problem. Currently, the best known algorithmic bound in the general hypergraph is $\Delta\lesssim q^{k/3}$ as established in \cite{JPV21,HSW21,FGWWY22}.\footnote{The notations $\lesssim$ and $\gtrsim$ hide low-order factors.}
In the case of linear hypergraphs, the bound can be improved to $q^{k/2-o(k)}$, as shown in \cite{FGW22a}.
From the hardness side, the approximate counting problem is known to be $\NP$-hard when $\Delta\geq 5\cdot q^{k/2}$ \cite{GGW22} in the general case when $q$ is even, but $\Delta\geq 2kq^k\log q+2q$ in the linear case \cite{GGW22}.

The problem of sampling hypergraph independent sets and $q$-colourings are both special cases of sampling solutions of constraint satisfaction problems (CSPs), and the computational results for these problems can be interpreted in the \emph{local lemma} regime. 
In this framework, each hyperedge is treated as a constraint associated with a ``bad event'' that this constraint fails when all vertices are independently assigned values uniformly at random. 
Suppose that each bad event occurs with probability at most $p$ and depends on at most $D$ other bad events. 
By noticing that $D=k\Delta-1$ in the hypergraph setting, the inapproximability conditions for counting both general hypergraph independent sets and colourings are unified as the form of $pD^2\gtrsim 1$. 
In contrast, the well-known Lov\'{a}sz local lemma \cite{EL75} provides a solution to \emph{any} CSP when $pD\lesssim 1$, and such a solution can be efficiently constructed using the Moser-Tardos Algorithm \cite{MT10}.
This confirms the presence of a quadratic gap between the complexity of approximate counting/sampling and searching. For general CSPs, the conjectured tractable regime is $pD^2\lesssim 1$, and there is a bunch of supporting evidence, see for example~\cite{WY23}. However, in light of the recent developments concerning linear hypergraphs, the tractable regime for CSPs with bounded overlap might extend further to $pD\lesssim 1$.

\section{Reduction from 2-spin systems}

Our reduction is based on the hardness of approximating the partition function of the \emph{$2$-spin system} on graphs.
A $2$-spin system on a graph $G=(V,E)$ is specified by an interaction matrix $\bm B$ and a vector ${\bm h}$ for the external field:
\begin{equation} \label{equ:2spin-interaction}
  {\bm B}=\begin{bmatrix}
    \beta & 1 \\ 1 & \gamma 
  \end{bmatrix}, \qquad 
  {\bm h}=\begin{bmatrix}
    \lambda \\ 1
  \end{bmatrix}, 
\end{equation}
where $\beta,\gamma,\lambda\geq 0$. 
The system is called \emph{anti-ferromagnetic}, if $\beta\gamma<1$. 
A configuration $\sigma:V\to\{0,1\}$ assigns each vertex $v\in V$ with a spin either $0$ or $1$. 
The \emph{weight} of a configuration $\sigma$ is defined by
\[
\wt(\sigma):=\lambda^{n_0(\sigma)}\beta^{m_{00}(\sigma)}\gamma^{m_{11}(\sigma)}
\]
where $n_0(\sigma)$ is the number of vertices assigned $0$ under $\sigma$, and $m_{00}(\sigma)$ (resp. $m_{11}(\sigma)$) is the number of edges whose both endpoints are assigned $0$ (resp. $1$) under $\sigma$.
The \emph{partition function} is defined by
\[
Z_{\beta,\gamma,\lambda}(G):=\sum_{\sigma}\wt(\sigma).
\]

The $2$-spin system we are interested in is specified by the following choices of parameters:
\begin{align}\label{eqn:parameters-specification}
    \beta=1, \qquad \gamma=1-\frac{1}{2^{k-2b}}, \qquad \lambda=2^b-1. 
\end{align}
The subscription in $Z_{\beta,\gamma,\lambda}$ is thus omitted as the parameters are now fixed. 

We now state the reduction. 
For any given $\Delta$-regular graph $G=(V,E)$, construct the hypergraph $H_G$ according to the following steps.  
\begin{itemize}
  \item[(T1)] Interpret the graph as a $2$-uniform hypergraph. 
  \item[(T2)] Replace each vertex with $b$ vertices. 
  \item[(T3)] For each hyperedge, insert another $k-2b$ vertices independently. 
\end{itemize}

Below is an example illustrating the reduction where $k=7$, $b=3$ and $\Delta=3$. 

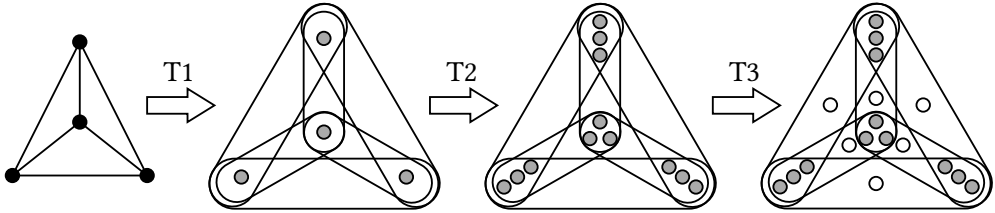
\begin{center}

  \tikzset{every picture/.style={line width=0.70pt}} 

  \begin{tikzpicture}[x=0.70pt,y=0.70pt,yscale=-1,xscale=1,scale=0.9]
  
  \draw    (28,156) -- (108,156) ;
  \draw    (28,156) -- (68,124) ;
  \draw    (108,156) -- (68,124) ;
  \draw    (68,76) -- (28,156) ;
  \draw    (68,76) -- (108,156) ;
  \draw    (68,76) -- (68,124) ;
  \draw  [fill={rgb, 255:red, 0; green, 0; blue, 0 }  ,fill opacity=1 ] (24,156) .. controls (24,153.79) and (25.79,152) .. (28,152) .. controls (30.21,152) and (32,153.79) .. (32,156) .. controls (32,158.21) and (30.21,160) .. (28,160) .. controls (25.79,160) and (24,158.21) .. (24,156) -- cycle ;
  \draw  [fill={rgb, 255:red, 0; green, 0; blue, 0 }  ,fill opacity=1 ] (64,76) .. controls (64,73.79) and (65.79,72) .. (68,72) .. controls (70.21,72) and (72,73.79) .. (72,76) .. controls (72,78.21) and (70.21,80) .. (68,80) .. controls (65.79,80) and (64,78.21) .. (64,76) -- cycle ;
  \draw  [fill={rgb, 255:red, 0; green, 0; blue, 0 }  ,fill opacity=1 ] (64,124) .. controls (64,121.79) and (65.79,120) .. (68,120) .. controls (70.21,120) and (72,121.79) .. (72,124) .. controls (72,126.21) and (70.21,128) .. (68,128) .. controls (65.79,128) and (64,126.21) .. (64,124) -- cycle ;
  \draw  [fill={rgb, 255:red, 0; green, 0; blue, 0 }  ,fill opacity=1 ] (104,156) .. controls (104,153.79) and (105.79,152) .. (108,152) .. controls (110.21,152) and (112,153.79) .. (112,156) .. controls (112,158.21) and (110.21,160) .. (108,160) .. controls (105.79,160) and (104,158.21) .. (104,156) -- cycle ;
  
  \draw  [color={rgb, 255:red, 0; green, 0; blue, 0 }  ,draw opacity=1 ][fill={rgb, 255:red, 170; green, 170; blue, 170 }  ,fill opacity=1 ] (536.98,63.78) .. controls (536.98,61.57) and (538.77,59.78) .. (540.98,59.78) .. controls (543.19,59.78) and (544.98,61.57) .. (544.98,63.78) .. controls (544.98,65.99) and (543.19,67.78) .. (540.98,67.78) .. controls (538.77,67.78) and (536.98,65.99) .. (536.98,63.78) -- cycle ;
  \draw  [color={rgb, 255:red, 0; green, 0; blue, 0 }  ,draw opacity=1 ][fill={rgb, 255:red, 170; green, 170; blue, 170 }  ,fill opacity=1 ] (536.98,73.78) .. controls (536.98,71.57) and (538.77,69.78) .. (540.98,69.78) .. controls (543.19,69.78) and (544.98,71.57) .. (544.98,73.78) .. controls (544.98,75.99) and (543.19,77.78) .. (540.98,77.78) .. controls (538.77,77.78) and (536.98,75.99) .. (536.98,73.78) -- cycle ;
  \draw  [color={rgb, 255:red, 0; green, 0; blue, 0 }  ,draw opacity=1 ][fill={rgb, 255:red, 170; green, 170; blue, 170 }  ,fill opacity=1 ] (536.98,83.78) .. controls (536.98,81.57) and (538.77,79.78) .. (540.98,79.78) .. controls (543.19,79.78) and (544.98,81.57) .. (544.98,83.78) .. controls (544.98,85.99) and (543.19,87.78) .. (540.98,87.78) .. controls (538.77,87.78) and (536.98,85.99) .. (536.98,83.78) -- cycle ;
  \draw  [color={rgb, 255:red, 0; green, 0; blue, 0 }  ,draw opacity=1 ][fill={rgb, 255:red, 170; green, 170; blue, 170 }  ,fill opacity=1 ] (530.98,133.78) .. controls (530.98,131.57) and (532.77,129.78) .. (534.98,129.78) .. controls (537.19,129.78) and (538.98,131.57) .. (538.98,133.78) .. controls (538.98,135.99) and (537.19,137.78) .. (534.98,137.78) .. controls (532.77,137.78) and (530.98,135.99) .. (530.98,133.78) -- cycle ;
  \draw  [color={rgb, 255:red, 0; green, 0; blue, 0 }  ,draw opacity=1 ][fill={rgb, 255:red, 170; green, 170; blue, 170 }  ,fill opacity=1 ] (542.98,133.78) .. controls (542.98,131.57) and (544.77,129.78) .. (546.98,129.78) .. controls (549.19,129.78) and (550.98,131.57) .. (550.98,133.78) .. controls (550.98,135.99) and (549.19,137.78) .. (546.98,137.78) .. controls (544.77,137.78) and (542.98,135.99) .. (542.98,133.78) -- cycle ;
  \draw  [color={rgb, 255:red, 0; green, 0; blue, 0 }  ,draw opacity=1 ][fill={rgb, 255:red, 170; green, 170; blue, 170 }  ,fill opacity=1 ] (536.98,123.78) .. controls (536.98,121.57) and (538.77,119.78) .. (540.98,119.78) .. controls (543.19,119.78) and (544.98,121.57) .. (544.98,123.78) .. controls (544.98,125.99) and (543.19,127.78) .. (540.98,127.78) .. controls (538.77,127.78) and (536.98,125.99) .. (536.98,123.78) -- cycle ;
  \draw  [color={rgb, 255:red, 0; green, 0; blue, 0 }  ,draw opacity=1 ][fill={rgb, 255:red, 170; green, 170; blue, 170 }  ,fill opacity=1 ] (495.98,150.78) .. controls (495.98,148.57) and (497.77,146.78) .. (499.98,146.78) .. controls (502.19,146.78) and (503.98,148.57) .. (503.98,150.78) .. controls (503.98,152.99) and (502.19,154.78) .. (499.98,154.78) .. controls (497.77,154.78) and (495.98,152.99) .. (495.98,150.78) -- cycle ;
  \draw  [color={rgb, 255:red, 0; green, 0; blue, 0 }  ,draw opacity=1 ][fill={rgb, 255:red, 170; green, 170; blue, 170 }  ,fill opacity=1 ] (487.98,156.78) .. controls (487.98,154.57) and (489.77,152.78) .. (491.98,152.78) .. controls (494.19,152.78) and (495.98,154.57) .. (495.98,156.78) .. controls (495.98,158.99) and (494.19,160.78) .. (491.98,160.78) .. controls (489.77,160.78) and (487.98,158.99) .. (487.98,156.78) -- cycle ;
  \draw  [color={rgb, 255:red, 0; green, 0; blue, 0 }  ,draw opacity=1 ][fill={rgb, 255:red, 170; green, 170; blue, 170 }  ,fill opacity=1 ] (479.98,162.78) .. controls (479.98,160.57) and (481.77,158.78) .. (483.98,158.78) .. controls (486.19,158.78) and (487.98,160.57) .. (487.98,162.78) .. controls (487.98,164.99) and (486.19,166.78) .. (483.98,166.78) .. controls (481.77,166.78) and (479.98,164.99) .. (479.98,162.78) -- cycle ;
  \draw  [color={rgb, 255:red, 0; green, 0; blue, 0 }  ,draw opacity=1 ][fill={rgb, 255:red, 170; green, 170; blue, 170 }  ,fill opacity=1 ] (593.98,162.78) .. controls (593.98,160.57) and (595.77,158.78) .. (597.98,158.78) .. controls (600.19,158.78) and (601.98,160.57) .. (601.98,162.78) .. controls (601.98,164.99) and (600.19,166.78) .. (597.98,166.78) .. controls (595.77,166.78) and (593.98,164.99) .. (593.98,162.78) -- cycle ;
  \draw  [color={rgb, 255:red, 0; green, 0; blue, 0 }  ,draw opacity=1 ][fill={rgb, 255:red, 170; green, 170; blue, 170 }  ,fill opacity=1 ] (577.98,150.78) .. controls (577.98,148.57) and (579.77,146.78) .. (581.98,146.78) .. controls (584.19,146.78) and (585.98,148.57) .. (585.98,150.78) .. controls (585.98,152.99) and (584.19,154.78) .. (581.98,154.78) .. controls (579.77,154.78) and (577.98,152.99) .. (577.98,150.78) -- cycle ;
  \draw  [color={rgb, 255:red, 0; green, 0; blue, 0 }  ,draw opacity=1 ][fill={rgb, 255:red, 170; green, 170; blue, 170 }  ,fill opacity=1 ] (585.98,156.78) .. controls (585.98,154.57) and (587.77,152.78) .. (589.98,152.78) .. controls (592.19,152.78) and (593.98,154.57) .. (593.98,156.78) .. controls (593.98,158.99) and (592.19,160.78) .. (589.98,160.78) .. controls (587.77,160.78) and (585.98,158.99) .. (585.98,156.78) -- cycle ;
  \draw   (528.98,69.78) .. controls (528.98,63.15) and (534.35,57.78) .. (540.98,57.78) -- (540.98,57.78) .. controls (547.61,57.78) and (552.98,63.15) .. (552.98,69.78) -- (552.98,129.78) .. controls (552.98,136.41) and (547.61,141.78) .. (540.98,141.78) -- (540.98,141.78) .. controls (534.35,141.78) and (528.98,136.41) .. (528.98,129.78) -- cycle ;
  \draw   (535.18,119.39) .. controls (540.92,116.08) and (548.26,118.05) .. (551.57,123.79) -- (551.57,123.79) .. controls (554.88,129.53) and (552.92,136.87) .. (547.18,140.18) -- (495.22,170.18) .. controls (489.48,173.49) and (482.14,171.53) .. (478.82,165.79) -- (478.82,165.79) .. controls (475.51,160.05) and (477.48,152.71) .. (483.22,149.39) -- cycle ;
  \draw   (599.18,149.79) .. controls (604.92,153.11) and (606.88,160.45) .. (603.57,166.19) -- (603.57,166.19) .. controls (600.26,171.93) and (592.92,173.89) .. (587.18,170.58) -- (535.22,140.58) .. controls (529.48,137.27) and (527.51,129.93) .. (530.82,124.19) -- (530.82,124.19) .. controls (534.14,118.45) and (541.48,116.48) .. (547.22,119.79) -- cycle ;
  \draw   (472.98,160.78) .. controls (472.98,152.5) and (479.7,145.78) .. (487.98,145.78) -- (593.98,145.78) .. controls (602.27,145.78) and (608.98,152.5) .. (608.98,160.78) -- (608.98,160.78) .. controls (608.98,169.06) and (602.27,175.78) .. (593.98,175.78) -- (487.98,175.78) .. controls (479.7,175.78) and (472.98,169.06) .. (472.98,160.78) -- cycle ;
  \draw   (547.84,54.84) .. controls (554.38,58.61) and (556.62,66.98) .. (552.85,73.51) -- (498.52,167.61) .. controls (494.74,174.15) and (486.38,176.39) .. (479.84,172.62) -- (479.84,172.62) .. controls (473.3,168.84) and (471.06,160.48) .. (474.84,153.94) -- (529.16,59.84) .. controls (532.94,53.3) and (541.3,51.06) .. (547.84,54.84) -- cycle ;
  \draw   (602.75,172.83) .. controls (596,176.73) and (587.37,174.42) .. (583.48,167.67) -- (529.58,74.32) .. controls (525.69,67.58) and (528,58.95) .. (534.75,55.05) -- (534.75,55.05) .. controls (541.49,51.16) and (550.12,53.47) .. (554.01,60.22) -- (607.91,153.56) .. controls (611.8,160.31) and (609.49,168.94) .. (602.75,172.83) -- cycle ;
  \draw   (536.98,109.78) .. controls (536.98,107.57) and (538.77,105.78) .. (540.98,105.78) .. controls (543.19,105.78) and (544.98,107.57) .. (544.98,109.78) .. controls (544.98,111.99) and (543.19,113.78) .. (540.98,113.78) .. controls (538.77,113.78) and (536.98,111.99) .. (536.98,109.78) -- cycle ;
  \draw   (520.98,137.78) .. controls (520.98,135.57) and (522.77,133.78) .. (524.98,133.78) .. controls (527.19,133.78) and (528.98,135.57) .. (528.98,137.78) .. controls (528.98,139.99) and (527.19,141.78) .. (524.98,141.78) .. controls (522.77,141.78) and (520.98,139.99) .. (520.98,137.78) -- cycle ;
  \draw   (552.98,137.78) .. controls (552.98,135.57) and (554.77,133.78) .. (556.98,133.78) .. controls (559.19,133.78) and (560.98,135.57) .. (560.98,137.78) .. controls (560.98,139.99) and (559.19,141.78) .. (556.98,141.78) .. controls (554.77,141.78) and (552.98,139.99) .. (552.98,137.78) -- cycle ;
  \draw   (509.84,113.73) .. controls (509.84,111.52) and (511.63,109.73) .. (513.84,109.73) .. controls (516.05,109.73) and (517.84,111.52) .. (517.84,113.73) .. controls (517.84,115.94) and (516.05,117.73) .. (513.84,117.73) .. controls (511.63,117.73) and (509.84,115.94) .. (509.84,113.73) -- cycle ;
  \draw   (564.98,113.78) .. controls (564.98,111.57) and (566.77,109.78) .. (568.98,109.78) .. controls (571.19,109.78) and (572.98,111.57) .. (572.98,113.78) .. controls (572.98,115.99) and (571.19,117.78) .. (568.98,117.78) .. controls (566.77,117.78) and (564.98,115.99) .. (564.98,113.78) -- cycle ;
  \draw   (536.98,160.78) .. controls (536.98,158.57) and (538.77,156.78) .. (540.98,156.78) .. controls (543.19,156.78) and (544.98,158.57) .. (544.98,160.78) .. controls (544.98,162.99) and (543.19,164.78) .. (540.98,164.78) .. controls (538.77,164.78) and (536.98,162.99) .. (536.98,160.78) -- cycle ;
  
  \draw  [color={rgb, 255:red, 0; green, 0; blue, 0 }  ,draw opacity=1 ][fill={rgb, 255:red, 170; green, 170; blue, 170 }  ,fill opacity=1 ] (372.98,63.89) .. controls (372.98,61.68) and (374.77,59.89) .. (376.98,59.89) .. controls (379.19,59.89) and (380.98,61.68) .. (380.98,63.89) .. controls (380.98,66.1) and (379.19,67.89) .. (376.98,67.89) .. controls (374.77,67.89) and (372.98,66.1) .. (372.98,63.89) -- cycle ;
  \draw  [color={rgb, 255:red, 0; green, 0; blue, 0 }  ,draw opacity=1 ][fill={rgb, 255:red, 170; green, 170; blue, 170 }  ,fill opacity=1 ] (372.98,73.89) .. controls (372.98,71.68) and (374.77,69.89) .. (376.98,69.89) .. controls (379.19,69.89) and (380.98,71.68) .. (380.98,73.89) .. controls (380.98,76.1) and (379.19,77.89) .. (376.98,77.89) .. controls (374.77,77.89) and (372.98,76.1) .. (372.98,73.89) -- cycle ;
  \draw  [color={rgb, 255:red, 0; green, 0; blue, 0 }  ,draw opacity=1 ][fill={rgb, 255:red, 170; green, 170; blue, 170 }  ,fill opacity=1 ] (372.98,83.89) .. controls (372.98,81.68) and (374.77,79.89) .. (376.98,79.89) .. controls (379.19,79.89) and (380.98,81.68) .. (380.98,83.89) .. controls (380.98,86.1) and (379.19,87.89) .. (376.98,87.89) .. controls (374.77,87.89) and (372.98,86.1) .. (372.98,83.89) -- cycle ;
  \draw  [color={rgb, 255:red, 0; green, 0; blue, 0 }  ,draw opacity=1 ][fill={rgb, 255:red, 170; green, 170; blue, 170 }  ,fill opacity=1 ] (366.98,133.89) .. controls (366.98,131.68) and (368.77,129.89) .. (370.98,129.89) .. controls (373.19,129.89) and (374.98,131.68) .. (374.98,133.89) .. controls (374.98,136.1) and (373.19,137.89) .. (370.98,137.89) .. controls (368.77,137.89) and (366.98,136.1) .. (366.98,133.89) -- cycle ;
  \draw  [color={rgb, 255:red, 0; green, 0; blue, 0 }  ,draw opacity=1 ][fill={rgb, 255:red, 170; green, 170; blue, 170 }  ,fill opacity=1 ] (378.98,133.89) .. controls (378.98,131.68) and (380.77,129.89) .. (382.98,129.89) .. controls (385.19,129.89) and (386.98,131.68) .. (386.98,133.89) .. controls (386.98,136.1) and (385.19,137.89) .. (382.98,137.89) .. controls (380.77,137.89) and (378.98,136.1) .. (378.98,133.89) -- cycle ;
  \draw  [color={rgb, 255:red, 0; green, 0; blue, 0 }  ,draw opacity=1 ][fill={rgb, 255:red, 170; green, 170; blue, 170 }  ,fill opacity=1 ] (372.98,123.89) .. controls (372.98,121.68) and (374.77,119.89) .. (376.98,119.89) .. controls (379.19,119.89) and (380.98,121.68) .. (380.98,123.89) .. controls (380.98,126.1) and (379.19,127.89) .. (376.98,127.89) .. controls (374.77,127.89) and (372.98,126.1) .. (372.98,123.89) -- cycle ;
  \draw  [color={rgb, 255:red, 0; green, 0; blue, 0 }  ,draw opacity=1 ][fill={rgb, 255:red, 170; green, 170; blue, 170 }  ,fill opacity=1 ] (331.98,150.89) .. controls (331.98,148.68) and (333.77,146.89) .. (335.98,146.89) .. controls (338.19,146.89) and (339.98,148.68) .. (339.98,150.89) .. controls (339.98,153.1) and (338.19,154.89) .. (335.98,154.89) .. controls (333.77,154.89) and (331.98,153.1) .. (331.98,150.89) -- cycle ;
  \draw  [color={rgb, 255:red, 0; green, 0; blue, 0 }  ,draw opacity=1 ][fill={rgb, 255:red, 170; green, 170; blue, 170 }  ,fill opacity=1 ] (323.98,156.89) .. controls (323.98,154.68) and (325.77,152.89) .. (327.98,152.89) .. controls (330.19,152.89) and (331.98,154.68) .. (331.98,156.89) .. controls (331.98,159.1) and (330.19,160.89) .. (327.98,160.89) .. controls (325.77,160.89) and (323.98,159.1) .. (323.98,156.89) -- cycle ;
  \draw  [color={rgb, 255:red, 0; green, 0; blue, 0 }  ,draw opacity=1 ][fill={rgb, 255:red, 170; green, 170; blue, 170 }  ,fill opacity=1 ] (315.98,162.89) .. controls (315.98,160.68) and (317.77,158.89) .. (319.98,158.89) .. controls (322.19,158.89) and (323.98,160.68) .. (323.98,162.89) .. controls (323.98,165.1) and (322.19,166.89) .. (319.98,166.89) .. controls (317.77,166.89) and (315.98,165.1) .. (315.98,162.89) -- cycle ;
  \draw  [color={rgb, 255:red, 0; green, 0; blue, 0 }  ,draw opacity=1 ][fill={rgb, 255:red, 170; green, 170; blue, 170 }  ,fill opacity=1 ] (429.98,162.89) .. controls (429.98,160.68) and (431.77,158.89) .. (433.98,158.89) .. controls (436.19,158.89) and (437.98,160.68) .. (437.98,162.89) .. controls (437.98,165.1) and (436.19,166.89) .. (433.98,166.89) .. controls (431.77,166.89) and (429.98,165.1) .. (429.98,162.89) -- cycle ;
  \draw  [color={rgb, 255:red, 0; green, 0; blue, 0 }  ,draw opacity=1 ][fill={rgb, 255:red, 170; green, 170; blue, 170 }  ,fill opacity=1 ] (413.98,150.89) .. controls (413.98,148.68) and (415.77,146.89) .. (417.98,146.89) .. controls (420.19,146.89) and (421.98,148.68) .. (421.98,150.89) .. controls (421.98,153.1) and (420.19,154.89) .. (417.98,154.89) .. controls (415.77,154.89) and (413.98,153.1) .. (413.98,150.89) -- cycle ;
  \draw  [color={rgb, 255:red, 0; green, 0; blue, 0 }  ,draw opacity=1 ][fill={rgb, 255:red, 170; green, 170; blue, 170 }  ,fill opacity=1 ] (421.98,156.89) .. controls (421.98,154.68) and (423.77,152.89) .. (425.98,152.89) .. controls (428.19,152.89) and (429.98,154.68) .. (429.98,156.89) .. controls (429.98,159.1) and (428.19,160.89) .. (425.98,160.89) .. controls (423.77,160.89) and (421.98,159.1) .. (421.98,156.89) -- cycle ;
  \draw   (364.98,69.89) .. controls (364.98,63.27) and (370.35,57.89) .. (376.98,57.89) -- (376.98,57.89) .. controls (383.61,57.89) and (388.98,63.27) .. (388.98,69.89) -- (388.98,129.89) .. controls (388.98,136.52) and (383.61,141.89) .. (376.98,141.89) -- (376.98,141.89) .. controls (370.35,141.89) and (364.98,136.52) .. (364.98,129.89) -- cycle ;
  \draw   (371.18,119.51) .. controls (376.92,116.2) and (384.26,118.16) .. (387.57,123.9) -- (387.57,123.9) .. controls (390.88,129.64) and (388.92,136.98) .. (383.18,140.29) -- (331.22,170.29) .. controls (325.48,173.61) and (318.14,171.64) .. (314.82,165.9) -- (314.82,165.9) .. controls (311.51,160.16) and (313.48,152.82) .. (319.22,149.51) -- cycle ;
  \draw   (435.18,149.91) .. controls (440.92,153.22) and (442.88,160.56) .. (439.57,166.3) -- (439.57,166.3) .. controls (436.26,172.04) and (428.92,174.01) .. (423.18,170.69) -- (371.22,140.69) .. controls (365.48,137.38) and (363.51,130.04) .. (366.82,124.3) -- (366.82,124.3) .. controls (370.14,118.56) and (377.48,116.6) .. (383.22,119.91) -- cycle ;
  \draw   (308.98,160.89) .. controls (308.98,152.61) and (315.7,145.89) .. (323.98,145.89) -- (429.98,145.89) .. controls (438.27,145.89) and (444.98,152.61) .. (444.98,160.89) -- (444.98,160.89) .. controls (444.98,169.18) and (438.27,175.89) .. (429.98,175.89) -- (323.98,175.89) .. controls (315.7,175.89) and (308.98,169.18) .. (308.98,160.89) -- cycle ;
  \draw   (383.84,54.95) .. controls (390.38,58.73) and (392.62,67.09) .. (388.85,73.63) -- (334.52,167.73) .. controls (330.74,174.27) and (322.38,176.51) .. (315.84,172.73) -- (315.84,172.73) .. controls (309.3,168.95) and (307.06,160.59) .. (310.84,154.05) -- (365.16,59.96) .. controls (368.94,53.42) and (377.3,51.18) .. (383.84,54.95) -- cycle ;
  \draw   (438.75,172.95) .. controls (432,176.84) and (423.37,174.53) .. (419.48,167.78) -- (365.58,74.44) .. controls (361.69,67.69) and (364,59.06) .. (370.75,55.17) -- (370.75,55.17) .. controls (377.49,51.27) and (386.12,53.58) .. (390.01,60.33) -- (443.91,153.68) .. controls (447.8,160.42) and (445.49,169.05) .. (438.75,172.95) -- cycle ;
  
  \draw  [color={rgb, 255:red, 0; green, 0; blue, 0 }  ,draw opacity=1 ][fill={rgb, 255:red, 170; green, 170; blue, 170 }  ,fill opacity=1 ] (208.98,73.89) .. controls (208.98,71.68) and (210.77,69.89) .. (212.98,69.89) .. controls (215.19,69.89) and (216.98,71.68) .. (216.98,73.89) .. controls (216.98,76.1) and (215.19,77.89) .. (212.98,77.89) .. controls (210.77,77.89) and (208.98,76.1) .. (208.98,73.89) -- cycle ;
  \draw  [color={rgb, 255:red, 0; green, 0; blue, 0 }  ,draw opacity=1 ][fill={rgb, 255:red, 170; green, 170; blue, 170 }  ,fill opacity=1 ] (209,130) .. controls (209,127.79) and (210.79,126) .. (213,126) .. controls (215.21,126) and (217,127.79) .. (217,130) .. controls (217,132.21) and (215.21,134) .. (213,134) .. controls (210.79,134) and (209,132.21) .. (209,130) -- cycle ;
  \draw  [color={rgb, 255:red, 0; green, 0; blue, 0 }  ,draw opacity=1 ][fill={rgb, 255:red, 170; green, 170; blue, 170 }  ,fill opacity=1 ] (159.98,156.89) .. controls (159.98,154.68) and (161.77,152.89) .. (163.98,152.89) .. controls (166.19,152.89) and (167.98,154.68) .. (167.98,156.89) .. controls (167.98,159.1) and (166.19,160.89) .. (163.98,160.89) .. controls (161.77,160.89) and (159.98,159.1) .. (159.98,156.89) -- cycle ;
  \draw  [color={rgb, 255:red, 0; green, 0; blue, 0 }  ,draw opacity=1 ][fill={rgb, 255:red, 170; green, 170; blue, 170 }  ,fill opacity=1 ] (257.98,156.89) .. controls (257.98,154.68) and (259.77,152.89) .. (261.98,152.89) .. controls (264.19,152.89) and (265.98,154.68) .. (265.98,156.89) .. controls (265.98,159.1) and (264.19,160.89) .. (261.98,160.89) .. controls (259.77,160.89) and (257.98,159.1) .. (257.98,156.89) -- cycle ;
  \draw   (200.98,69.89) .. controls (200.98,63.27) and (206.35,57.89) .. (212.98,57.89) -- (212.98,57.89) .. controls (219.61,57.89) and (224.98,63.27) .. (224.98,69.89) -- (224.98,129.89) .. controls (224.98,136.52) and (219.61,141.89) .. (212.98,141.89) -- (212.98,141.89) .. controls (206.35,141.89) and (200.98,136.52) .. (200.98,129.89) -- cycle ;
  \draw   (207.18,119.51) .. controls (212.92,116.2) and (220.26,118.16) .. (223.57,123.9) -- (223.57,123.9) .. controls (226.88,129.64) and (224.92,136.98) .. (219.18,140.29) -- (167.22,170.29) .. controls (161.48,173.61) and (154.14,171.64) .. (150.82,165.9) -- (150.82,165.9) .. controls (147.51,160.16) and (149.48,152.82) .. (155.22,149.51) -- cycle ;
  \draw   (271.18,149.91) .. controls (276.92,153.22) and (278.88,160.56) .. (275.57,166.3) -- (275.57,166.3) .. controls (272.26,172.04) and (264.92,174.01) .. (259.18,170.69) -- (207.22,140.69) .. controls (201.48,137.38) and (199.51,130.04) .. (202.82,124.3) -- (202.82,124.3) .. controls (206.14,118.56) and (213.48,116.6) .. (219.22,119.91) -- cycle ;
  \draw   (144.98,160.89) .. controls (144.98,152.61) and (151.7,145.89) .. (159.98,145.89) -- (265.98,145.89) .. controls (274.27,145.89) and (280.98,152.61) .. (280.98,160.89) -- (280.98,160.89) .. controls (280.98,169.18) and (274.27,175.89) .. (265.98,175.89) -- (159.98,175.89) .. controls (151.7,175.89) and (144.98,169.18) .. (144.98,160.89) -- cycle ;
  \draw   (219.84,54.95) .. controls (226.38,58.73) and (228.62,67.09) .. (224.85,73.63) -- (170.52,167.73) .. controls (166.74,174.27) and (158.38,176.51) .. (151.84,172.73) -- (151.84,172.73) .. controls (145.3,168.95) and (143.06,160.59) .. (146.84,154.05) -- (201.16,59.96) .. controls (204.94,53.42) and (213.3,51.18) .. (219.84,54.95) -- cycle ;
  \draw   (274.75,172.95) .. controls (268,176.84) and (259.37,174.53) .. (255.48,167.78) -- (201.58,74.44) .. controls (197.69,67.69) and (200,59.06) .. (206.75,55.17) -- (206.75,55.17) .. controls (213.49,51.27) and (222.12,53.58) .. (226.01,60.33) -- (279.91,153.68) .. controls (283.8,160.42) and (281.49,169.05) .. (274.75,172.95) -- cycle ;
  
  \draw   (108,109) -- (132,109) -- (132,104) -- (148,114) -- (132,124) -- (132,119) -- (108,119) -- cycle ;
  
  \draw   (276,109) -- (300,109) -- (300,104) -- (316,114) -- (300,124) -- (300,119) -- (276,119) -- cycle ;
  
  \draw   (444,109) -- (468,109) -- (468,104) -- (484,114) -- (468,124) -- (468,119) -- (444,119) -- cycle ;

  \draw (116,86) node [anchor=north west][inner sep=0.75pt]   [align=left] {T1};
  \draw (284,86) node [anchor=north west][inner sep=0.75pt]   [align=left] {T2};
  \draw (452,86) node [anchor=north west][inner sep=0.75pt]   [align=left] {T3};

  \end{tikzpicture}
  
\end{center}

It is immediate to verify that $H_G$ is $k$-uniform, has overlap $b$ and maximum degree $\Delta$. 
Let $\+{I}(H_G)$ be the set of independent sets of $H_G$. 

\begin{lemma} \label{lem:count-equivalence}
  For any $\Delta$-regular graph $G=(V,E)$ and the constructed hypergraph $H_G$, it holds that $|\mathcal{I}(H_G)|=2^{|E|(k-2b)}Z(G)$. 
\end{lemma}

\begin{proof}
We define the following partition over all the independent sets $\mathcal{I}(H_G)
=\biguplus_{\sigma}\mathcal{S}(\sigma)$ in the hypergraph $H_G$, where $\sigma$ ranges over all configurations of the $2$-spin system.
For any vertex $v\in V$, let $B_v$ be the set of constructed vertices in $H_G$ corresponding to $v$ as in step (T2) of the construction.
Given an independent set $I\in\mathcal{I}(H_G)$, the part $\mathcal{S}(\sigma)$ that $I$ falls into is given by, for any $v\in V$,
\begin{itemize}
  \item $\sigma(v)=0$, if $|B_v\cap I|\leq b-1$; 
  \item $\sigma(v)=1$, if $|B_v\cap I|=b$ (namely, $B_v\subseteq I$). 
\end{itemize}
This is a partition because each $I\in\mathcal{I}(H_G)$ falls into exactly one part. 
Next, we demonstrate that $|S(\sigma)|=2^{|E|(k-2b)}\wt(\sigma)$, which can be immediately used to establish the lemma.  In subsequent discussions, a partial configuration of a subset $S\subseteq V(H_G)$ refers to a subset of vertices $S'\subseteq S$ that is included in the (weak) independent set of $H_G$. A partial configuration $S'$ of $S$ is considered feasible if there exists an independent set $I\in\mathcal{I}(H_G)$ such that $I\cap S=S'$.



\begin{itemize}
  \item Consider the vertices constructed in (T2). 
  \begin{itemize}
    \item For each $v\in V$ such that $\sigma(v)=0$, there are $2^b-1$ feasible partial configurations of $B_v$. 
    \item For each $v\in V$ such that $\sigma(v)=1$, there is just one feasible partial configuration of $B_v$. 
  \end{itemize}
  \item Consider the vertices constructed in (T3). 
  \begin{itemize}
    \item For each edge $e$ such that both its endpoints take spin $1$, the rest $k-2b$ vertices of the corresponding hyperedge cannot be in an independent set together, so there are $2^{k-2b}-1$ feasible partial configurations. 
    \item For any other edge, the corresponding $k-2b$ vertices are free to be included in an independent set, so there are $2^{k-2b}$ feasible partial configurations. 
  \end{itemize}
\end{itemize}
In all, this gives
\[
  |S(\sigma)|=\left(2^{k-2b}-1\right)^{m_{11}(\sigma)}\left(2^{k-2b}\right)^{|E|-m_{11}(\sigma)}\left(2^b-1\right)^{n_0(\sigma)}=2^{|E|(k-2b)}\wt(\sigma). \qedhere
\]
\end{proof}

Our goal then boils down to showing the inapproximability of the constructed $2$-spin system. 
To establish this, we invoke the following celebrated result by Sly and Sun \cite{SS14}, which connects the so-called non-uniqueness property of any general anti-ferromagnetic $2$-spin system with computational hardness. 
Denote by $\mathbb{T}_{\Delta}$ the infinite $\Delta$-regular tree, and by $\hat{\mathbb{T}}_{\Delta}$ the infinite $(\Delta-1)$-ary tree. 
\begin{theorem}[\cite{SS14}] \label{lem:ss14}
  For any nondegenerate homogeneous anti-ferromagnetic $2$-spin system specified by the interaction matrix $\bm{B}$ and the external field $\bm{h}$ on $\Delta$-regular graphs that lies in the $\mathbb{T}_{\Delta}$ non-uniqueness region, the partition function is $\NP$-hard to approximate, even within a factor of $2^{cn}$ for some constant $c>0$ depending on $\Delta$ and the spin system. 
\end{theorem}

We remark that uniqueness/non-uniqueness regions for $\mathbb{T}_{\Delta}$ coincide with those for $\hat{\mathbb{T}}_{\Delta}$. 
However, $(\Delta-1)$-ary trees are more convenient to handle, so we move to $\hat{\mathbb{T}}_{\Delta}$ onwards. 
It is known that the $\hat{\mathbb{T}}_{\Delta}$ (non-)uniqueness region can be characterised by the solutions of the standard tree recursion formula on the ratio of the marginal probability of the root induced by the Gibbs measure.
The following lemma, originally due to Martinelli, Sinclair and Weitz \cite[Section 6.2]{martinelli2007fast}, characterises these solutions. 
\begin{lemma}[{\cite[Lemma 7]{GSV16}}] \label{lem:msw}
For $\Delta\geq 3$ and an anti-ferromagnetic $2$-spin system specified by (\ref{equ:2spin-interaction}),
consider the system of equations
\begin{equation}\label{equ:one-step}
  x=R(y), \qquad 
  y=R(x), \qquad
  \text{ where } R(z):=\lambda\left(\frac{\beta z+1}{z+\gamma}\right)^{\Delta-1}
  \text{ and }x,y\geq 0.
\end{equation}
Then, 
\begin{itemize} 
  \item in the $\hat{\mathbb{T}}_{\Delta}$ uniqueness region, the system has a unique solution $(Q^\times,Q^\times)$; 
  \item in the $\hat{\mathbb{T}}_{\Delta}$ non-uniqueness region, the system has three solutions $(Q^+,Q^-), (Q^\times,Q^\times), (Q^-,Q^+)$ where $Q^+>Q^\times>Q^-$. 
\end{itemize}
\end{lemma}

$R(z)$ in \eqref{equ:one-step} is called the \emph{one-step recursion}. 
For any pair of solutions $(x,y)$ to \eqref{equ:one-step}, we have that $R(R(x))=x$ and $R(R(y))=y$. Therefore, both $x$ and $y$ are fixed points of the \emph{two-step recursion} $R(R(z))$. 
Let $d:=\Delta-1$. 
Using the above lemma, it suffices to show that the two-step recursion has three fixed points $Q^+>Q^\times>Q^-$ in order to establish non-uniqueness. 
Equivalently, in the regime of parameters specified in \Cref{thm:main}, we need to show that the following function, obtained by plugging the parameters specified in~\eqref{eqn:parameters-specification} into the two-step recursion, has three distinct zeros in $(0,+\infty)$: 
\begin{equation}
  f(z):=(2^b-1)\left(1+\frac{1}{2^{k-2b}(2^b-1)\left(1+\frac{1}{2^{k-2b}z+2^{k-2b}-1}\right)^d+2^{k-2b}-1}\right)^{d}-z. 
\end{equation}
Since the solution $Q^\times$ is the unique fixed point of the one-step recursion, that is, $Q^\times=R(Q^\times)$, it is also helpful to consider the function
\begin{equation}
  g(z):=(2^b-1)\left(1+\frac{1}{2^{k-2b}z+2^{k-2b}-1}\right)^d-z
\end{equation} obtained by plugging the parameters specified in~\eqref{eqn:parameters-specification} into the one-step recursion.
The following lemma is sufficient to derive our main theorem.  
\begin{lemma} \label{lem:analytic}
  Assume integers $k\geq 2$, $1\leq b\leq k/2$ and $d=5\cdot 2^{k-b}$. 
  Define $z^*:=d/2^{k-2b}=5\cdot 2^b$. 
  Then $f(z^*)>0$ and $g(z^*)<0$. 
\end{lemma}

\begin{proof}[Proof of \Cref{thm:main}]
Note that $g(0)>0$ and $\lim_{z\to +\infty} g(z)=-\infty$. 
By $g(z^*)<0$, we know that the unique zero of $g$, which is $Q^\times$, is smaller than $z^*$. 
On the other hand, by $f(z^*)>0$ and $\lim_{z\to +\infty} f(z)=-\infty$, there is a zero of $f$ on $(z^*,+\infty)$, and it cannot be $Q^\times$. 
This establishes non-uniqueness due to \Cref{lem:msw}. 
$\NP$-hardness then follows after \Cref{lem:ss14}. 
\end{proof}

In the proof of \Cref{lem:analytic}, the following standard inequality is useful. 
\begin{equation} \label{equ:exp}
  \exp\{x\}>\left(1+\frac{x}{y}\right)^y>\exp\left\{\frac{xy}{x+y}\right\}\qquad\text{for all }x,y>0. 
\end{equation}

\begin{proof}[Proof of \Cref{lem:analytic}]
The $g(z^*)$ part is due to a straightforward estimation:
\begin{align*}
g(z^*)
\leq (2^b-1)\left(1+\frac{1}{5\cdot 2^{k-b}}\right)^{5\cdot 2^{k-b}}-5\cdot 2^b < (2^b-1)\mathrm{e}-5\cdot 2^b < 0. 
\end{align*}

When $k=2$, $b$ can only take $1$, and a direct calculation gives $f(z^*)>16.0$ in this case. For $k\geq 3$, we have
\begin{align*}
f(z^*)=&(2^b-1)\left(1+\frac{1}{2^{k-2b}(2^b-1)\left(1+\frac{1}{d+2^{k-2b}-1}\right)^{d}+2^{k-2b}-1}\right)^{d}-5\cdot 2^b\\
\ge &(2^b-1)\left(1+\frac{1}{2^{k-2b}(2^b-1)\left(1+\frac{1}{d}\right)^{d}+2^{k-2b}-1}\right)^{d}-5\cdot 2^b\\
> &(2^b-1)\left(1+\frac{1}{2^{k-2b}(2^b-1)\mathrm{e}+2^{k-2b}-1}\right)^{d}-5\cdot 2^b\\
> &(2^b-1)\left(1+\frac{1}{2^{k-2b}(2^b-1)\mathrm{e}+2^{k-2b}}\right)^{d}-5\cdot 2^b\\
>&(2^b-1)\exp\left\{\frac{5\cdot 2^{b+k}}{2^k+2^{2b}+2^{k}(2^b-1)\mathrm{e}}\right\}-5\cdot 2^b. \tag{By (\ref{equ:exp}) and $d=5\cdot 2^{k-b}$}
\end{align*}

Note that the fraction in $\exp\{\cdot\}$ is monotone increasing with respect to $k$. We can further analyse this by considering two cases.
\begin{itemize}
  \item Fix $b\geq 2$. The whole term is minimised at $k=2b$.
  Plugging this in, we further get
  \[
  f(z^*)>(2^b-1)\exp\left\{\frac{5\cdot 2^b}{2+(2^b-1)\mathrm{e}}\right\}-5\cdot 2^b =: h(b).
  \]
  Now it suffices to show $h(b)>0$ for $b\geq 2$. 
  If $b=2$, we have $h(2)>1.5$.
  If $b\geq 3$, then
  \[
  h(b)> (2^b-1)\exp\left\{\frac{5\cdot 2^b}{\mathrm{e}\cdot 2^b}\right\}-5\cdot 2^b>1.29\cdot 2^b-6.29>0.
  \]
  \item Fix $b=1$. The whole term is minimised at $k=3$ (we assume $k\geq 3$), and the value is at least $0.7$. 
\end{itemize}

Combining all these facts, the proof is complete.
\end{proof}

\section*{Acknowledgement}

Jiaheng Wang has received funding from the European Research Council (ERC) under the European Union's Horizon 2020 research and innovation programme (grant agreement No. 947778), and an Informatics Global PhD Scholarship at The University of Edinburgh.





\bibliographystyle{alpha} 
\bibliography{linear-indset-hardness}

\end{document}